\documentclass[reqno,12pt]{amsart}
\usepackage{amsmath,epsfig}
\usepackage{amssymb}

\newtheoremstyle{rem}{1.3ex}{1.3ex}{\rmfamily}{}
{\itshape\rmfamily}{}{1.5ex}{}

\newtheorem{theorem}{Theorem}[section]
\newtheorem{lemma}[theorem]{Lemma}

\newtheorem{corollary}[theorem] {Corollary}

\theoremstyle{definition}
\newtheorem{example}[theorem] {Example}
\newtheorem{remark}[theorem] {Remark}

\renewcommand{\section}{\secdef\sct\sect}
\newcommand{\sct}[2][default]{\refstepcounter{section}
\setcounter{equation}{0}
\vspace{0.5cm}
\centerline{ \large
\scshape \arabic{section}.\ #1}
\vspace{0.3cm}}
\newcommand{\sect}[1]{
\vspace{0.5cm}
\centerline{\large\scshape #1}
\vspace{0.3cm}}

\renewcommand{\subsection}{\secdef \subsct\sbsect}
\newcommand{\subsct}[2][default]{\refstepcounter{subsection}
\nopagebreak
\vspace{0.5\baselineskip}
{\flushleft\bf \arabic{section}.\arabic{subsection}~\bf #1  }
\nopagebreak}
\newcommand{\sbsect}[1]{\vspace{0.1cm}\noindent
{\bf #1}\vspace{0.1cm}}

\def\rr{\mathbb R}

\def\nn{\mathbb N}

\def\phi{\varphi }

\def\cali{{\mathcal I}}
\def\calj{{\mathcal J}}

\def\calm{{\mathcal M}}
\def\caln{{\mathcal N}}

\def\cals{{\mathcal S}}

\def\cM{{\mathcal M}}

\newcommand{\supp}{{\operatorname {supp}}}

\newcommand{\R}     {\mathbb{R}}

\newcommand{\N}     {\mathbb{N}}
\renewcommand{\P}   {\mathbb{P}}

\newcommand{\E}     {\mathbb{E}}

\def\1{{\mathchoice {1\mskip-4mu\mathrm l}
                    {1\mskip-4mu\mathrm l}
                    {1\mskip-4.5mu\mathrm l} {1\mskip-5mu\mathrm l}}}

\begin{document}

\title[Large deviations for disordered bosons]{\large
Large deviations for\\\vspace{2mm}disordered bosons and\\\vspace{5mm}multiple orthogonal polynomial ensembles}

\author[Peter Eichelsbacher, Jens Sommerauer and Michael Stolz]{} 
\maketitle
\thispagestyle{empty}
\vspace{0.2cm}

\centerline{\sc Peter Eichelsbacher\footnote{Corresponding author: Ruhr-Universit\"at Bochum, Fakult\"at f\"ur Mathematik,
NA 3/66, D-44780 Bochum, Germany, {\tt peter.eichelsbacher@ruhr-uni-bochum.de}
\\All authors have been supported by Deutsche Forschungsgemeinschaft via SFB/TR 12}, Jens Sommerauer, Michael Stolz}
\centerline{\sc (Ruhr-Universit\"at Bochum)}


\vspace{2 cm}

\begin{quote}
{\small {\bf Abstract:} }
We prove a large deviations principle
for the empirical measures of a class of biorthogonal and multiple orthogonal
polynomial ensembles that includes
biorthogonal Laguerre, Jacobi and Hermite ensembles,
the matrix model of Lueck, Sommers and Zirnbauer for disordered bosons, the Stieltjes-Wigert matrix model of 
Chern-Simons theory, and Angelesco ensembles.
\end{quote}




\newpage
\setcounter{section}{0}

\section{Introduction}
{\it Orthogonal polynomial ensembles} are given by joint probability densities of the form
\begin{equation} \label{ortho}
p(x_1, \ldots, x_n) = c_n \prod_{i < j} (x_i-x_j)^2 \prod_{i=1}^n w(x_i)
\end{equation}
with a positive weight function $w$. The classical examples, where $w$ is the weight of Hermite, Laguerre, or Jacobi
polynomials, naturally arise in random matrix theory: as joint eigenvalue distribution of the Gaussian Unitary Ensemble
(Hermite), of Wishart matrices (Laguerre), or of random projectors (Jacobi). 

In \cite{Borodin:1999}, Borodin studied more general joint probability densities of the form 
\begin{equation} \label{biorth}
p(x_1, \ldots, x_n) = c_n \prod_{i < j} (x_i-x_j)(x_i^{\theta} - x_j^{\theta}) \prod_{i=1}^n w(x_i),
\end{equation}
where $\theta$ is a fixed positive number. As pointed out in \cite{Borodin:1999}, these are related to biorthogonal polynomials, and will be referred to as {\it biorthogonal ensembles}. 

This note is mainly motivated by a special biorthogonal ensemble that arises from a random matrix model for {\it disordered bosons} that was proposed by Lueck, Sommers,
and Zirnbauer in \cite{Lueck/Sommers/Zirnbauer:2006}. The model amounts to the product of a Wishart matrix and the fundamental matrix of the standard symplectic form, and on the level of eigenvalues, interpreted as characteristic frequencies of disordered quasi-particles, one obtains the joint density
\begin{equation} \label{densitybosonic}
{\tilde q}_{n,\alpha}(x_1,\ldots, x_{n})=c_n \prod_{i=1}^{n}x_i ^{\alpha}\mathrm{e}^{-\tau x_i} 
\prod_{1\leq i<j \leq n}|x_i-x_j|\left|x_i^{2}-x_j^{2}\right|1_{[0,\infty)^n}(x_1,\ldots, x_{n})
\end{equation}
with $\alpha \in \nn$. In \cite{Lueck/Sommers/Zirnbauer:2006} is was shown
that the correlation functions of the frequencies in the bulk of the spectrum are in the Gaussian Unitary Ensemble universality class, yet a novel
scaling behaviour is found at the low frequency end of the spectrum. Other applications of biorthogonal ensembles to physics are discussed in \cite{Muttalib:1995} and \cite{Tierz:2010}.  In the latter reference, the motivation comes from matrix models for Chern-Simons theory.

The aim of this note is to complement the Lueck-Sommers-Zirnbauer results by a large deviations 
principle for the empirical measure of the characteristic frequencies. For any sequence of random numbers $x_1, \ldots, x_n$ we denote
by
$$
L_n : = \frac 1n \sum_{i=1}^n \delta_{x_i}
$$
the empirical distribution or {\it empirical measure} of these values (a random probability measure on $\Bbb R$). Define
the mean empirical measure $\bar{L}_n = \E L_n$ by the relation $\langle \bar{L}_n, f \rangle = \E \langle L_n, f \rangle$ for all
continuous and bounded functions $f : \Bbb R \to \Bbb R$. One result in \cite{Lueck/Sommers/Zirnbauer:2006} is that the sequence
$(\bar{L}_n)_n$ converges weakly to a probability measure on $\Bbb R$ with Lebesgue density
\begin{equation} \label{bosonrho}
\varrho_{\infty} (t) : = \frac{1}{2\pi} (t/b)^{-1/3} \bigl( (1+ \sqrt{1-t^2/b^2})^{1/3} -(1- \sqrt{1-t^2/b^2})^{1/3} \bigr), 
\end{equation}
for $0<t\le b:= 3\sqrt{3}$.
One consequence of the large deviations principle that will be proven in what follows is that this statement can
be improved: we will show that the empirical measures $(L_n)_n$ themselves converge weakly, in probability, to $\varrho_{\infty}$.

Actually, we will study large deviations principles for empirical measures 
in a broader framework that encompasses not only ensembles like \eqref{biorth}, but also takes care
of weight functions $w_n$, depending on $n$, and determinantal parts like
$ \prod_{i < j} |x_i-x_j|^{\beta} (x_i^{\theta} - x_j^{\theta})$ for any $\beta >0$. We will also obtain large deviations results
for {\it multiple orthogonal polynomial} ensembles (see \cite{Kuijlaars:2010}). The main observation is that while
these large deviations results do depend on the explicit formulae for the joint distributions of the eigenvalues like \eqref{ortho} and \eqref{biorth},
they are independent of the determinantal structure of the correlation functions, and we need not invoke the theory of orthogonal or biorthogonal
or multiple orthogonal polynomials.

In \cite{BenArous/Guionnet:1997} (see also \cite{Zeitounibook}), Ben Arous and Guionnet have obtained a large deviations principle (LDP) for the empirical measure of eigenvalues from the
Wigner-Dyson ensembles GOE/GUE/GSE in the space ${\mathcal M}_1(\Bbb R)$ of probability measures on the Borel sets of $\Bbb R$, endowed with the weak topology, with speed
$n^2$ and good rate function (GOE case)
$$
I(\mu) = \frac 14 \int x^2 \, \mu(dx) + \frac 12 \int \int \log |x-y|^{-1} \mu(dx) \mu(dy) - \frac 38,
$$
whose unique minimiser is the semicircle distribution. Recall that a family of probability measures $(\mu_{\varepsilon})_{\varepsilon >0}$ on a topological space
$X$ is said to obey a large deviations principle (LDP) with speed $\varepsilon^{-1}$ and good rate function $I: X \to [0,\infty]$ if $I$ is lower semi-continuous
and has compact level sets $N_L := \{x \in X: I(x) \leq L \}$, for every $L \in[0, \infty)$, and
$$
\liminf_{\varepsilon \to 0} \varepsilon \log \mu_{\varepsilon}(G) \geq - \inf_{x \in G} I(x)
$$
for every open $G \subseteq X$ and
$$
\limsup_{\varepsilon \to 0} \varepsilon \log \mu_{\varepsilon}(A) \leq - \inf_{x \in A} I(x)
$$
for every closed $A \subseteq X$.

This has been generalized in \cite{Eichelsbacher/Stolz:2006} to joint densities of the form 
\begin{equation} \label{hamiltonian}
q_n(x_1,\ldots, x_{p(n)})=\frac{1}{Z_n}\prod_{i=1}^{p(n)}w_n(x_i)^n \prod_{1\leq i<j \leq p(n)}
\left|x_i^{\theta}-x_j^{\theta}\right|^{\beta}, 
\end{equation}
with $\theta \in \N$, $\beta >0$, partition function $Z_n$ and continuous weight functions $w_n:\R\rightarrow\R_0^+$. 
This framework takes care of the needs of mesoscopic physics in that it subsumes matrix versions of all classical
symmetric spaces (see \cite{Heinzner/Huckleberry/Zirnbauer:2005}). With regard to the symmetry classification of 
disordered fermionic systems provided in the last reference,  
let us remark that an analogous classification for the case of bosons is not completely understood, see, however,
the discussion in \cite[Section 4]{Zirnbauer:2010}. Moreover, a $d$-dimensional generalisation of the random matrix
model for disordered bosons in \cite{Lueck/Sommers/Zirnbauer:2006} was studied recently in \cite{Schmittner/Zirnbauer:2010}.

The note is organised as follows. Section 2 is devoted to the formulation of a large deviations principle for the empirical measures
$(L_n)_n$ of the eigenvalues of generalised biorthogonal matrix ensembles. The examples include the random matrix
model of disordered bosons in \cite{Lueck/Sommers/Zirnbauer:2006}, the Stieltjes-Wigert ensembles in \cite{Tierz:2010} as well
as biorthogonal Jacobi-, Laguerre- and Hermite ensembles considered in \cite{Borodin:1999}. In Section 3 we formulate large deviations
principles for a special multiple orthogonal ensemble, the Angelesco ensemble, see \cite{Kuijlaars:2010}.
In Sections 4 and 5 we present the proofs of our large deviations principles.


\section{Large deviations for biorthogonal ensembles and beyond}

In this section, we will derive a LDP for the bosonic ensemble, where 
the density of the joint distribution of the eigenvalues is of form \eqref{densitybosonic}.
Obviously, \eqref{densitybosonic} is a special case of the density
\begin{equation}\label{density}
q_n(x_1,\ldots, x_{p(n)})=\frac{1}{Z_n}\prod_{i=1}^{p(n)}w_n(x_i)^n \prod_{1\leq i<j \leq p(n)}|x_i-x_j|
\left|x_i^{\theta}-x_j^{\theta}\right|1_{\Sigma^{p(n)}}(x_1,\ldots, x_{p(n)}),
\end{equation}
with $\theta \in \N$, partition function $Z_n$ and continuous weight functions $w_n:\R\rightarrow\R_0^+$. For $\theta$ even, 
$\Sigma$ is a closed subset of $[0,\infty)$ while for $\theta$ odd, $\Sigma$ is a closed subset of $\R$. The 
sequence $(p(n))_n$ must satisfy
$$
\lim_{n\rightarrow \infty} \frac{p(n)}{n}=\kappa \in (0,\infty).
$$
Note that for $p(n)=n$ and weight functions $w$ independent of $n$, \eqref{density} subsumes the density for the eigenvalue distribution 
for biorthogonal ensembles as introduced in \cite{Borodin:1999}. For $\theta =1$, $p(n)=n$ and $w(x)=\mathrm{e}^{-\frac{1}{2}x^2}$ we 
also recover the classical GUE.

Throughout the whole section, we write ${\mathcal N}(f)$ for the set of zeros of a function $f:\R\rightarrow \R$ and we assume that the sequence of weight functions $(w_n)_n$ satisfies the following:
\begin{itemize}
\item[(a1)] there exists a continuous function $w: \Sigma \to [0,
\infty )$ such that 
\begin{itemize}
\item $\# \caln(w) < \infty,\ \caln(w_n) \subseteq \caln(w)$ for 
large $n$. (a1.1) 
\item As $n \to \infty$, $w_n$ converges to $w$, and $\log w_n$ to
$\log w$ uniformly on compact sets. (a1.2)  
\item  $\log w$ is Lipschitz on compact sets away from $\caln(w)$. (a1.3)
\end{itemize}
\item[(a2)] If $\Sigma$ is unbounded, then there exists $n_0 \in
\nn$ such that
$$ \lim_{x \to \pm \infty} |x|^{(\theta+1) (\kappa + \epsilon)} \sup_{n \ge n_0} w_n(x) = 0$$
for some fixed $\epsilon > 0$.
\end{itemize}
We will study the asymptotic behaviour of the empirical distribution $L_{n}(x)$ of $x$ 
for $x = (x_1, \ldots, x_{p(n)}) \in \Sigma^{p(n)}$, which is defined as
$$
L_{n}(x) := \frac{1}{p(n)} \sum_{j=1}^{p(n)} \delta_{x_j},
$$
Let the space of probability measures on the Borel sets of $\Sigma$ be denoted by $\cM_1(\Sigma)$ and let 
$Q_n$ denote the joint distribution of random variables $(X_1,\ldots, X_{p(n)})$ with density $q_n$. The main theorem now reads as follows:

\begin{theorem} \label{ldp}
The sequence 
$\left( Q_n \circ L_{n}^{-1}\right)_n$ satisfies a
LDP on $\mathcal M_1(\Sigma)$ with respect to the weak topology
with speed $n^2$ and good rate function
\begin{eqnarray} \label{rate}
I(\mu) & = & \frac{\kappa^2}{2} \int \int \left\{\log \left|x^{\theta} - y^{\theta}\right|^{-1}+
\log |x - y|^{-1}\right\} \, \mu(dx) \, \mu(dy) \\
&& -  \kappa \int \log w(x) \, \mu(dx) + c, \nonumber
\end{eqnarray}
where $\mu \in \mathcal M_1(\Sigma)$ and
\begin{eqnarray} \label{qv12}
c :=  \lim_{n \to \infty} \frac{1}{n^2} \log Z_n 
&=&-\inf_{\mu \in \mathcal{M}_1(\Sigma)} \left\{ \frac{\kappa^2}{2} \int \int \left\{\log \left|x^{\theta}
 - y^{\theta}\right|^{-1}+\log |x - y|^{-1}\right\} \, \mu(dx) \, \mu(dy) \right.\nonumber \\
 && \left. \qquad \qquad-  \kappa \int \log w(x) \, \mu(dx) \right\}
< \infty.
\end{eqnarray} 
\end{theorem}

\begin{corollary}\label{SLLN}
Whenever the joint density of the eigenvalues is as in \eqref{density} and the rate function $I$ has a unique minimiser $\mu^{\ast}$, we obtain under (a1) and (a2) a strong law of large numbers,
$$
\P(L_n \stackrel{weak}{\longrightarrow} \mu^{\ast})=1.
$$ 
Moreover for any closed $A \subset \cM_1(\Sigma)$ with $A \cap \{ \nu \in \cM_1(\Sigma): I(\nu)=0 \}= \emptyset$, then
for all $n \geq n_0$
$$
Q_n(L_n \in A) \leq \exp \bigl( - n^2 \inf_{\nu \in A} I(\nu)/2 \bigr).
$$
\end{corollary}

\begin{proof}{}
Employing the upper bound of the LDP, the strong law of large numbers follows via an application of Borel-Cantelli's
lemma, see \cite[Thm. II.6.3]{Ellis:LargeDeviations}. 
\end{proof}

\begin{example}[Disordered bosons]
Returning to the bosonic ensemble with density \eqref{densitybosonic}, we have as weight functions
$$
w_n(x)=x^{\frac{\alpha}{n}}e^{-\frac{\tau x}{n}}. 
$$ 
Now $\tau^{-1}$ is the variance of the independent and normally distributed random variables, that were used to 
construct the stability matrix $h$ for that ensemble, cf. \cite{Lueck/Sommers/Zirnbauer:2006}. We will take the variance $\tau^{-1}$ equal to 
$n^{-1}$. For a greater generality we take a sequence $(\alpha(n))_{n \in \N}$ with $\lim_{n \rightarrow \infty}\frac{\alpha(n)}{n}=\alpha>-1$. 
The case $\alpha =0$ incorporates having a constant sequence $(\alpha(n))_n$. Obviously, conditions (a1) and (a2) are met and we obtain 
for the ensemble of disordered bosons, that $(Q_n \circ  L_{n}^{-1})_n $ satisfies a LDP on $\mathcal M_1(\Sigma)$ with respect to the weak topology 
with speed $n^2$ and good rate function
\begin{eqnarray*} 
I(\mu)  =  -\frac{\kappa^2}{2} \int \int \log \left(\left|x^{2} - y^{2}\right||x - y|\right) \, \mu(dx) \, \mu(dy) \\
-\kappa \int \alpha \log x \,\,\mu(dx) + \kappa   \int  x \, \mu(dx) + \tilde{c}, 
\end{eqnarray*}
where $\mu \in \mathcal M_1(\Sigma)$ and
\begin{equation*} 
\tilde{c} :=  \lim_{n \to \infty} \frac{1}{n^2} \log \tilde{Z}_n < \infty.
\end{equation*}
In \cite{Lueck/Sommers/Zirnbauer:2006}, the authors proved the weak law of large numbers (WLLN)
$$
\E\left[\frac{1}{N} \sum_{i=1}^N 1_{\{\lambda_i\le x\}}\right] \rightarrow \int_{-\infty}^x \rho_{\infty}(t) dt \quad \mbox{for } 
N\rightarrow \infty,
$$
with $\varrho_{\infty}$ defined in \eqref{bosonrho}.
They also stated that the corresponding rate function $I$ is convex and thus, one can find a unique minimiser $\mu^{\ast}$ of $I$, 
$I(\mu^{\ast})=0$, with density $\rho_{\infty}$. Therefore the WLLN of \cite{Lueck/Sommers/Zirnbauer:2006} is extended into a strong one.
\end{example}

\begin{example}[Stieltjes-Wigert ensembles]
In the case of Stieltjes-Wigert ensembles form \cite{Tierz:2010} the weight function is $w(x) = e^{- c (\log (x))^2}$
with some constant $c$. Hence our Theorem applies and we obtain a LDP and a strong law of large numbers.
\end{example}

\begin{example}[Jacobi-, Laguerre- and Hermite biorthogonal ensembles]
In \cite{Borodin:1999}, Jacobi ensembles with $w(x)=x^{\alpha}$ on $(0,1)$ are studied. We consider the general case
$$
\Sigma = [0,1], \qquad w_n(x)= x^{\alpha/n} (1-x)^{\beta/n}, \qquad \alpha, \beta > -1.
$$
For $\theta=1$ these ensembles appear in the canonical correlation analysis, where the correlations coefficients turn out to be the square
root of the eigenvalues of a special kind of matrix. It is known that the eigenvalues of these matrices follow a joint distribution
of Jacobi type. Taking $\alpha, \beta >-1$ fixed and $\frac{p(n)}{n} \to \kappa$, the sequence $(Q_n \circ L_n^{-1})_n$ obeys a LDP with speed $n^2$ and rate function
(up to a constant)
\begin{equation} \label{grundform}
\frac{\kappa^2}{2} \int \int \bigl( \log|x^{\theta} - y^{\theta}|^{-1} +\log |x-y|^{-1} \bigr) \mu(dx)\mu(dy).
\end{equation}
For $\frac{\alpha(n)}{n} \to \alpha >-1$ and $\frac{\beta(n)}{n} \to \beta >-1$, the corresponding LDP holds true, and the rate function is (up to a constant)
\eqref{grundform} plus
$$
- \kappa \int \alpha \log x \mu(dx) - \kappa \int \beta \log(1-x) \mu(dx).
$$
Hermite ensembles can be considered with $w_n(x) = |x|^{\alpha/n} e^{- x^2/n}$ on $(-\infty,\infty)$ with $\alpha >-1$.
The case $\alpha=0$ is the classical Hermite weight. For constant $\alpha >-1$ we obtain the LDP, and the rate function is \eqref{grundform}
plus $\kappa \int x^2 \mu(dx)$ with the same $\frac{p(n)}{n} \to \kappa$. For $\frac{\alpha(n)}{n} \to \alpha >-1$, the rate is \eqref{grundform}
plus $\kappa \int x^2 \mu(dx) - \kappa \int \alpha \log x \mu(dx)$. Finally, the Laguerre biorthogonal ensembles are given by the weight function
$w_n(x)= x^{\alpha/n} e^{- x/n}$ on $(0, \infty)$ with $\alpha >-1$. It is straightforward to find the corresponding rate functions for
constant and $n$-dependent parameters. For $\theta=2$ the Laguerre case is corresponding to the matrix model of disordered bosons in \cite{Lueck/Sommers/Zirnbauer:2006}.
\end{example}

\begin{remark}
It is obvious how to generalise biorthogonal ensembles. Consider an ensembles of $n$ points on $(a,b) \subset \Bbb R$ with the joint probability
density of the form
\begin{equation} \label{detdet}
c \prod_{i=1}^n w_n(x_i)^n \,\, \det [\xi_i(x_j)]_{i,j=1}^n \,\, \det [\eta_i(x_j)]_{i,j=1}^n,
\end{equation}
where $\xi_i(x), \eta_i(x)$, $i \geq 1$, are some functions defined on $(a,b)$. In case $p(n)=n$ the density \eqref{density} is clearly
a special case of \eqref{detdet}, take $\xi_i(x) = x^{i-1}$ and $\eta_i(x) = x^{\theta(i-1)}$. It is known, see \cite{Desrosiers/Forrester:2008}
and references therein, that random matrix theory provides many instances of such biorthogonal structures, for example unitary invariant matrix ensembles or
unitary ensembles with an external source, see also \cite{Bleher/Kuijlaars:2004} and the next section. 
But large deviations principles for the corresponding empirical measures of $n$ points distributed
according to \eqref{detdet} are out of range with respect to the techniques we are using to prove Theorem \ref{ldp}.
\end{remark}

\section{Large deviations for multiple orthogonal ensembles}
Multiple orthogonal polynomials are a generalisation of orthogonal polynomials in which the orthogonality
is distributed among a number of orthogonality weights. They appear in random matrix theory in the form of special determinantal
point processes that are called multiple orthogonal polynomial (MOP) ensembles. In \cite{Kuijlaars:2010, Kuijlaars:2010b}
the appearance of MOP in a variety of random matrix models and models related with particles
following non-intersecting paths have been considered. 
To a finite number of weight functions $w_1, \ldots, w_p$ on $\Bbb R$ and a multi-index $\vec{n} =(n_1, \ldots, n_p) \in {\Bbb N}^p$ we
associate a monic polynomial $P_{\vec{n}}$ of degree $n:= |\vec{n}| := n_1 + \cdots + n_p$ such that
$$
\int_{-\infty}^{\infty} P_{\vec{n}}(x) x^k w_j(x) \, dx = 0, \quad \text{for} \,\, k=0, \ldots, n_j-1, \,\, j=1, \ldots, p.
$$
If $P_{\vec{n}}$ uniquely exists then it is called the multiple orthogonal polynomial (MOP) associated with the weights $w_1, \ldots, w_p$
and multi-index $\vec{n}$. In \cite{Kuijlaars:2010} the following result was presented. Assume that
\begin{equation} \label{condMOP}
\frac{1}{Z_n} \det [f_j(x_k)]_{j,k=1, \ldots,n} \biggl[ \prod_{1 \leq j < k \leq n} (x_k-x_j) \biggr]
\end{equation}
is a probability density function on ${\Bbb R}^n$, where the linear span of $f_1, \ldots, f_n$ is the same as the linear span
of $\{x^k w_j(x) | j=0, \ldots, n_j-1, \, j=1, \ldots, p\, \}$. Then the MOP exists and is given by
$$
P_{\vec{n}}(x) = \E \biggl[ \prod_{j=1}^n(x-x_j) \biggr],
$$
where the expectation is taken with respect to the p.d.f \eqref{condMOP}, which can be interpreted as the expectation
of the random polynomial $\prod_{j=1}^n (x-x_j)$ with roots $x_1, \ldots, x_n$ from a determinantal point process
on the real line. The p.d.f \eqref{condMOP} is called a {\it MOP ensemble}. It was first observed in \cite{Bleher/Kuijlaars:2004} that
random matrix models with an external source lead naturally to MOP ensembles. 

The weights $w_1, \ldots, w_p$ are an {\it Angelesco system} if there are disjoint intervals $\Gamma_1, \ldots, \Gamma_p \subset {\Bbb R}$, such
that $\supp (w_j) \subset \Gamma_j$, $j=1, \ldots, p$. In the Angelesco case, $\det [f_j(x_k)]_{j,k=1, \ldots,n}$ is of block form and results in
$$
\det [f_j(x_k)]_{j,k=1, \ldots,n} = \prod_{i=1}^p \biggl( \Delta(X^{(i)}) \prod_{k=1}^{n_i} w_i(x_k^{(i)}) \biggr)
$$
with $x_k^{(i)} := x_{N_{i-1}+k} \in \Gamma_i$, $N_i=\sum_{j=1}^i n_j$ (with $N_0=0$) and $X^{(i)} = (x_1^{(i)}, \ldots, x_{n_i}^{(i)})$ and
$$
\Delta(X) = \prod_{1 \leq j <k \leq n} (x_k-x_j) \quad \text{for} \,\, X=(x_1, \ldots,x_n),
$$
the Vandermonde determinant. Thus an Angelesco system gives rise to a MOP ensemble, the {\it Angelesco ensemble}, and the joint
p.d.f is
\begin{equation} \label{angelesco}
\frac{1}{Z_n} \prod_{i=1}^p \Delta(X^{(i)})^2 \prod_{1 \leq i < j \leq p} \Delta(X^{(i)}, X^{(j)}) \prod_{i=1}^p \prod_{k=1}^{n_i} w_i(x_k^{(i)}),
\end{equation}
where 
$$
\Delta(X,Y) := \prod_{k=1}^n \prod_{j=1}^m (x_k-y_j)
$$
for $X=(x_1, \ldots,x_n)$ and $Y=(y_1, \ldots, y_m)$.

We now consider the situation that $|\vec{n}| = n \to \infty$ and $n_j \to \infty$ for every $j=1, \ldots, p$ in such a way that
\begin{equation} \label{c1} 
\frac{n_j}{n} \to r_j \quad \text{for} \,\, j=1, \ldots, p
\end{equation}
with $0 < r_j <1$ and $\sum_{j=1}^p r_j =1$. Let us consider varying weights 
\begin{equation} \label{varying}
w_i(x)= e^{-n V_i(x)} 
\end{equation} 
for any $i=1, \ldots, p$. Denote by
\begin{equation} \label{n1}
L_j(x^{(j)}) := L_{j,n_j}(x^{(j)}) := \frac{1}{n_j} \sum_{k=1}^{n_j} \delta_{x_k^{(j)}}
\end{equation}
the $j$-{\it th empirical measure} of $x^{(j)}$ for every $j=1, \ldots, p$.

We will study the asymptotic behaviour of the empirical distribution vector $L_{n}(x)$ of $x$ 
for $x = (x_1, \ldots, x_{n}) \in \prod_{j=1}^p \Gamma_j^{n_j} =: \Gamma_{p,\vec{n}}$, which is defined as
$$
L_{n}(x) := (L_1(x^{(1)}), \ldots, L_p(x^{(p)})) 
$$
with $L_i(x^{(i)})$ defined as in \eqref{n1}. $L_n(x)$ is an element in $\times_{i=1}^p \cM_1(\Gamma_i^{n_i}) =: \cM_1(\Gamma_{p, \vec{n}})$.
Let $Q_n$ denote the joint distribution of random variables $(X^{(1)}, \ldots, X^{(p)})$ with density \eqref{angelesco}. 

\begin{theorem}[LDP for Angelesco ensembles] \label{ldpMOP}
Assume that every weight function $w_i$ in \eqref{varying} satisfies assumption (a1) and (a2) and assume that \eqref{c1}
is fulfilled. Then the sequence 
$\left( Q_n \circ L_{n}^{-1}\right)_n$ satisfies a
LDP on $\mathcal M_1(\Gamma_{p, \vec{n}})$ with respect to the weak topology
with speed $n^2$ and good rate function
\begin{eqnarray} \label{ratean}
I(\mu_1, \ldots, \mu_p) & = & \frac 12 \sum_{j=1}^p r_j^2 \int \int \log \left|x - y\right|^{-2} \mu_j(dx) \, \mu_j(dy) \\
&& + \sum_{j=1}^{p-1} \sum_{k=j+1}^p r_j r_k \int \int \log \left|x - y\right|^{-1} \mu_j(x) \mu_k(y) + \sum_{j=1}^p r_j \int V_j(x) \mu_j(x) + c, \nonumber
\end{eqnarray}
where $(\mu_1, \ldots, \mu_p) \in \mathcal M_1(\Gamma_{p,\vec{n}})$ and
$$
c :=  \lim_{n \to \infty} \frac{1}{n^2} \log Z_n.
$$
\end{theorem}

\begin{corollary}
Whenever the joint density of the eigenvalues is as in \eqref{angelesco} and the rate function $I$ has a unique minimiser $\mu^{\ast}$, 
we obtain under (a1) and (a2) a strong law of large numbers,
$$
\P(L_n \stackrel{weak}{\longrightarrow} \mu^{\ast})=1.
$$ 
\end{corollary}

\begin{remark}
Hence the vector of empirical measures of a random points from the Angelesco ensembles tends to the vector of nonrandom measures
almost surely weakly. This solves the question posed in \cite[Section 5.2]{Kuijlaars:2010}. Remark that recently, in \cite{Bloom:2011},
a strong law for Angelesco ensembles was established, applying the notion of Fekete points as well as the Bernstein-Markov inequality.
We obtain a full LDP.
We would also be able to consider {\it Nikishin ensembles} with $p \geq 2$ weights, see \cite{Kuijlaars:2010} and references therein.
This is because the determinantal structure of the joint density of the eigenvalues \cite[(4.14)]{Kuijlaars:2010} consists
of Vandermonde-like products and hence the techniques of our proof of Theorem \ref{ldp} can be adapted. Nikishin interaction
arises in the asymptotic analysis of eigenvalues of {\it banded Toeplitz matrices} as well as in a {\it two-matrix model}, 
see \cite[Section 5.4]{Kuijlaars:2010}.
\end{remark}

\section{Proofs}
Now we present the proof of Theorem \ref{ldp}. It is a generalisation of arguments used
for the proofs of \cite[Theorem 2.6.1]{Zeitounibook} and \cite[Theorem 4.1]{Eichelsbacher/Stolz:2006}.
The proof of the upper bound is quite similar to the proofs in the latter references. In the proof
of the lower bound, a coarse graining argument is more involved. To overcome the singularity of the
logarithm at certain points is the most delicate part. 

Let us define $F:\Sigma \times \Sigma \rightarrow {\overline \R}$ by
\begin{displaymath}
F(x,y) := - \frac{\kappa^2}{2} \left( \log |x -y| + \log |x^{\theta} -
y^{\theta}|\right) - \frac{\kappa}{2} \bigl( \log w(x) + \log w(y) \bigr),
\end{displaymath}
where we set 
$F(x,y)=\infty$ if $x^{\theta}=y^{\theta}$ or if $\{x,y\}\cap \mathcal{N}(w)\not= \emptyset$. 
Note that due to our definition of $\Sigma$, $x^{\theta}=y^{\theta}$ corresponds to $x=y$.
Let $F^M(x,y)$ denote the truncated version of $F(x,y)$, for $ M>0$: $F^M(x,y):= F(x,y)\wedge M$. 
Furthermore, we define the functions $F_n:\Sigma \times \Sigma \rightarrow {\overline \R}$,
$$
F_n(x,y) := - \frac{1}{2}\left(\frac{p(n)}{n}\right)^2 \left( \log |x -y| + \log \left|x^{\theta} -
y^{\theta}\right|\right) - \frac{p(n)}{2n} \bigl( \log w_n(x) + \log w_n(y) \bigr),
$$
again with $F_n(x,y)=\infty$ if $x^{\theta}=y^{\theta}$ or if $\{x,y\}\cap \mathcal{N}(w_n)\not= \emptyset$
and their truncated versions
$$
F^M_n(x,y):= F_n(x,y)\wedge M,\quad M>0.
$$
Observing that
\begin{eqnarray*}
n\sum_{i=1}^{p(n)}\log w_n(x_i)
=\frac{n}{p(n)} \sum_{1\le i<j\le p(n)}\left(\log w_n(x_i)+\log w_n(x_j) \right)+\frac{n}{p(n)}
\sum_{i=1}^{p(n)}\log w_n(x_i),
\end{eqnarray*}
we can deduce from \eqref{density} and the definition of $F_n$ the following identity,
where we abbreviate $1_{\Sigma^{p(n)}}(x_1,..,x_{p(n)})$ by $1_{\Sigma^{p(n)}}(x)$,
\begin{eqnarray}
q_n(x_1,\ldots, x_{p(n)})
=\frac{1}{Z_n}\exp \left\{-\frac{2n^2}{p(n)^2} \sum_{i<j}F_n(x_i,x_j)+\frac{n}{p(n)} \sum_{i=1}^{p(n)}
\log w_n(x_i)  \right\} 1_{\Sigma^{p(n)}}(x). 
\end{eqnarray}
One key ingredient for the proof of the upper bound will be the following lemma, which also provides that 
the rate function is well defined. 

\begin{lemma} \label{bounded}
~ 
\begin{itemize}
\item[(i)]  For any $M>0$, $F^M_n(x,y)$ converges to $F^M(x,y)$
uniformly as $n \to \infty$.

\item[(ii)] $F$ is bounded from below.
\end{itemize}
\end{lemma}

\begin{proof}{}
Since $\log |x-y| \leq \log(|x|+1) + \log(|y|+1)$ holds for any $x,y \in\R$, it
implies
\begin{eqnarray} \label{qv45}
F_n(x,y) 
\geq - \frac{p(n)}{2n} 
\left[ \log \left(
\left[(|x|+1)(|x^{\theta}|+1)\right]^{\frac{p(n)}{n}} \, w_n(x) \right) 
 +\log \left(\left[(|y|+1)(|y^{\theta}|+1)\right]^{\frac{p(n)}{n}} \, w_n(y) \right)\right]. \nonumber
\end{eqnarray}
We will show that $F^M=F_n^M=M$ on some specified sets and 
then deal with the complement of these sets.
We will start by observing that $\log \left( \left[(|x^{\theta}|+1)(|x|+1)\right]^{\frac{p(n)}{n}} \,  
w_n(x) \right)$ is bounded from above: In case of $x\in[-1,1]\cap \Sigma$ it turns out that  
$$
\log \left( \left[(|x^{\theta}|+1)(|x|+1)\right]^{\frac{p(n)}{n}} 
w_n(x) \right)\le \log \left( 4^{\frac{p(n)}{n}}\right) +\log(w_n(x)).
$$
For $M>0$ we can take $n_1$ as large, so that (a1.1) is applicable and $|\frac{p(n)}{n}-\kappa|<\epsilon$ and $|w_n(x)-w(x)|<\epsilon 
\ \forall\ n\ge n_1$ and $\forall\ x \in [-1,1]$.
For each $ \nu \in \mathcal{N}(w)$ a $\delta^{(1)}_{\nu,M}>0$ and $\forall\ x \in[-1,1]$ with $|x-\nu|<\delta^{(1)}_{\nu,M}$, we now find 
$\log \left( 4^{\frac{p(n)}{n}}\right) +\log(w_n(x))\leq \log(4^{\kappa+\epsilon})+\log(w(x))\leq -M$. 
Whereas for $|x|\geq 1$, we observe that for $n\ge n_1$ as above the following holds 
\begin{eqnarray}\label{bound}
\log \left( \left[(|x^{\theta}|+1)(|x|+1)\right]^{\frac{p(n)}{n}}  w_n(x) \right)
\leq \log \left(4^{\kappa +\epsilon}|x|^{(\theta+1)(\kappa +\epsilon)} \sup_{m\ge n} w_m(x) \right). 
\end{eqnarray}
Assumption (a2) provides that for each $M>0$ there exists $n_2\in \N$ and $R_M>0$ such that for 
$|x| \ge R_M$ and all $n\ge n_2$ we have that
$$
\log \left(4^{\kappa +\epsilon}|x|^{(\theta+1)(\kappa +\epsilon)} \sup_{m\ge n} w_m(x) \right)\leq -M.
$$
In case of $|x|\le R_M$ reasoning as in the first part yields for each $M>0$ and $\nu \in \mathcal{N}(w)$ 
the existence of $\delta^{(2)}_{\nu,M}>0$, such that for $x$ with $|x-\nu|<\delta_{\nu,M}^{(2)}$ we also get 
$\log \left(4^{\frac{p(n)}{n}}|x|^{(\theta+1)\frac{p(n)}{n}} \,w_n(x) \right)\leq-M$.
Thus, for each $M > 0$ there exists $n_0:=\max\{n_1,n_2\}$, 
$\delta_{M,\nu}:=\min\{\delta^{(1)}_{M,\nu},\delta^{(2)}_{M,\nu}\}>0$ and $R_M>0$
such that $F_n(x, y) \ge M$ holds for all $n\ge n_0$ on
\begin{equation} \label{upperlemma}
 A_M := \{ |x| \vee |y| > R_M \} \cup \bigcup_{\nu \in \caln(w) }  \{ |x - \nu| \wedge |y - \nu| <
\delta_{\nu, M} \}.
\end{equation}  
From \eqref{upperlemma} we see that $A_M^c$ is compact. Moreover, we find a constant $C_M>0$, depending on $M$, such that $F_n\ge M$ and $F\ge M\ 
\forall \ x,y \in B_M:=\{(x,y) \in \R^2: |x-y|<C_M \}$, and therefore, a constant $0<\tilde{C}_{M}<\infty$ exists with 
$\max_{x,y \in B_M^c}\{-\log|x-y|, -\log|x^{\theta}-y^{\theta}|\} <\tilde{C}_{M}$. 
This implies the existence of yet another constant $0<C_{1,M}<\infty$ such that 
$$
\max_{x,y \in B_M^c\cap A_M^c}\{|\log|x-y||, |\log|x^{\theta}-y^{\theta}||\} <C_{1,M}.
$$
Remember that $x^{\theta}=y^{\theta}$ only holds if $x=y$. 
Furthermore, we introduce the notation $\parallel f \parallel _{\infty}^D :=\sup_{x \in D} |f(x)|$,
which restricts the supremum norm to a set $D$. 
The continuity of $w_n$ and the compactness of $A_M^c$ yields that 
$\parallel \log w_n(x)+\log w_n(y)\parallel_{\infty}^{A_M^c}<C_{2,M}$, for a constant $0<C_{2,M}<\infty$.
Thus,  for any given $\eta >0$ we find on $D:=A_M^c\setminus B_M$ a $\tilde{n} \in \N$, such that $\forall\ n\ge \tilde{n}$
\begin{eqnarray}
\parallel F_n(x,y)-F(x,y) \parallel_{\infty}^{D}
\le\eta. \nonumber 
\end{eqnarray} 
Hence, we have established the uniform convergence of $F_n^M$ to $F^M$ on $A_M^c\setminus B_M$, whereas on $A_M \cup 
B_M$ we have  
$F_n^M=M$. That $F^M=M$ holds on $A_M \cup B_M$ can be shown along the same lines and (i) is proven.
Besides $F \ge M$ on $A_M\cup B_M$, we also know that $F$ is real-valued and continuous on the compact set $A^c_M \cap B_M^c$ 
and therefore $F$ is bounded from below, which yields (ii).
\end{proof}

Before proving the LDP for $(Q_n \circ L_n^{-1})_n$, we consider first the finite measure 
$$
P_n:=Z_nQ_n,
$$
where $Z_n$ is the partition function of the Lebesgue density belonging to $Q_n$, see \eqref{density}, and 
derive a LDP for $(P_n \circ L_n^{-1})$. 
We set
$$
H(\mu) := \int F d\mu^{\otimes 2},\ H^M(\mu) := \int F^M d\mu^{\otimes
2},
$$
and obtain well defined maps on $\calm_1(\Sigma)$ (due to Lemma {bounded}). We claim the following:
\begin{lemma}\label{ratefct}
$H$ is a good rate function
that governs the LDP for $(P_n \circ L_n^{-1})_n$. 
\end{lemma}

\begin{proof}{}
Exactly as the proof of  \cite[Le. 3.5]{Eichelsbacher/Stolz:2006}, since Lemma \ref{bounded} assures that $F$ has the necessary properties to mimic the proof of \cite[Le. 3.5]{Eichelsbacher/Stolz:2006}.
\end{proof} 

Moving on, we will show that $(P_n \circ L_n^{-1})_n$ fulfils a weak LDP. 

\subsection{Proof of the upper bound}
The main obstacle will be to overcome the singularities of the function $F$. 
We start by observing that  for $\Delta :=\{(x,y) \in \Sigma ^2: x=y\} $,
\begin{eqnarray}\label{empmeasure}
L_n\otimes L_n (\Delta)=\frac{1}{p(n)}
\end{eqnarray}
holds $m_{\R^{p(n)}}$-almost surely, because of the a.s. distinct eigenvalues under the product Lebesgue measure
$m_{\R^{p(n)}}$.
From \eqref{empmeasure} it follows that $Q_n$-a.s. 
\begin{equation}
\int\int_{x \not= y} F_n^M(x,y)L_n(dx)L_n(dy)=\int\int F_n^M(x,y)L_n(dx)L_n(dy)-\frac{M}{p(n)}.
\end{equation}
For a Borel set $A$ in ${\mathcal M}_1(\Sigma)$ let us define $\tilde{A}$ as follows
${\tilde A}:=\left\{ x\in \Sigma^{p(n)}:L_n(x)\in A\right\}$.
Due to the symmetry of $F_n$ in its arguments we find
$
2\sum_{i<j}F_n(x_i,x_j)=\sum_{i\not=j}F_n(x_i,x_j)$, 
and since $m_{\R^{p(n)}}$ is the product Lebesgue measure on $\R^{p(n)}$, we 
obtain via H\"older's inequality
\begin{eqnarray} \label{hoelder}
P_n \bigl( L_{n} \in A \bigr) 
& \leq & \biggl(  \int  \exp \left\{\frac{2 n}{p(n)} \log w_{n}(t)
\right\} \, {\rm m}_{\rr}(dt) \biggr)^{\frac{p(n)}{{\tt 2}}} \nonumber \\ 
& &\left( \int_{\tilde{A}} \exp \left\{ - \frac{2 n^2}{p(n)^2}
\sum_{i \neq j} F_n^M(x_i, x_j) \right\} {\rm m}_{\rr^{p(n)}}(dx)
\right)^{1/2}
 \nonumber\\
&=& {\rm (I)} \times {\rm (II)}. \nonumber
\end{eqnarray}
First, we show $\lim_{n\rightarrow \infty} \frac{1}{n^2}\log {\rm (I)}=0$.
Since $w_n$ is continuous, hence bounded on compact sets, the key observation is that (a2) implies
\begin{eqnarray*}
\int_{[-K,K]^c}  w_n(t)^{\frac{2n}{p(n)}}\ {\rm m}_{\rr}(dt)&=&\int_{[-K,K]^c}   w_n(t)^{\frac{2n}{p(n)}} 
\left(\frac{|t|^{(\theta +1)(\kappa  +\epsilon)}}{|t|^{(\theta +1)(\kappa  +\epsilon)}}
\right)^{\frac{2n}{p(n)}} {\rm m}_{\rr}(dt) \\
&\le& \int_{[-K,K]^c} \frac{1}{|x|^{1+\eta}}{\rm m}_{\rr}(dt),\nonumber 
\end{eqnarray*}
for suitable $K, \eta > 0$ and $n$ large enough. Concerning (II), we find that for any $M>0$
\begin{eqnarray}
{\rm (II)}  &=& \left\{ \int_{\tilde{A}} \exp \left( - 2 n^2 \left( L_{n}(x)^{\otimes 2}(F_n^M) - 
\frac{M}{p(n)} \right) \right) {\rm m}_{\rr^{p(n)}}(dx)
\right\}^{1/2} \nonumber \\
&\le& \left\{ \exp \left( - 2 n^2 \left( \inf_{\mu \in A} \mu^{\otimes 2}(F_n^M) - \frac{M}{p(n)} \right) 
\right) \right\}^{1/2} \nonumber \\
&=& \exp\left(- n^2  \inf_{\mu \in A} \mu^{\otimes 2}(F_n^M)\right) \exp\left( \frac{M n^2}{p(n)}\right).
\end{eqnarray}
The first part of Lemma \ref{bounded} yields,
$$
\lim_{n \to \infty} \biggl( \inf_{\mu \in A} \mu^{\otimes
2}(F_n^M) \biggr) = \inf_{\mu \in A} \mu^{\otimes 2}(F^M).
$$
We have thus shown that for any Borel set $A \subset \mathcal
M_1(\Sigma)$ one has
\begin{equation} \label{qv18}
\limsup_{n \to \infty} \frac{1}{n^2} \log P_n \bigl( L_{n} \in A
\bigr) \leq - \inf_{\mu \in A} \int \int F^M(x,y) \, \mu(dx)
\mu(dy).
\end{equation}
The inequality \eqref{qv18} allows us to proof the following lemma.
\begin{lemma}\label{exptight}
$(P_n \circ L_n^{-1})_n$ is exponentially tight.
\end{lemma}
\begin{proof}{}
Very similar to \cite[p.12]{Eichelsbacher/Stolz:2006} so that we omit the proof here.
\end{proof}
Now we choose for $A$ the set $B(\mu, \delta)=\{\nu \in {\mathcal M}_1(\Sigma): d(\nu,\mu)\leq 2\delta\}$ with 
$\delta>0$ and 
$$ d(\mu,\nu)=\sup \left| \int fd\mu -\int fd\nu \right|$$ where the supremum is taken over 
all Lipschitz functions $f$ for which the sum of the Lipschitz constant $l_f$ and of the uniform bound 
$\parallel f \parallel_{\infty}$ is less than or equal to 1. 
That distance is compatible with the weak topology, see \cite[p.356]{Dembo/Zeitouni:LargeDeviations}. 
Since $\mu \mapsto H^M(\mu)$ is weakly continuous, from \eqref{qv18} we obtain for any $\mu \in \mathcal 
M_1(\Sigma)$,
\begin{eqnarray*}
\inf_{\delta \to 0} \limsup_{n \to \infty} \frac{1}{n^2} \log P_n
\bigl( L_{n} \in B(\mu, \delta) \bigr) &\stackrel{\eqref{qv18}}{\leq}& -\inf_{\delta \to 0} 
\inf_{\nu \in B(\mu,\delta)} \int \int F^M(x,y)\ \nu(dx)\nu(dy)
\nonumber \\ &\leq& - \int \int F^M(x,y) \,\mu(dx)  \mu(dy) =-H^M(\mu).
\end{eqnarray*}
Finally, letting $M$ go to infinity, we obtain the following upper
bound,
\begin{equation}\label{upperbound}
\inf_{\delta \to 0} \limsup_{n \to \infty} \frac{1}{n^2} \log P_n
\bigl( L_{n} \in B(\mu, \delta) \bigr) \leq - H(\mu). 
\end{equation}

\subsection{Proof of the lower bound}
Turning to the lower bound, we take $B(\mu, \delta)$ as above and for every $\mu \in \mathcal M_1(\Sigma)$ 
we will show a lower bound,
\begin{eqnarray} \label{lowbound}
\inf_{\delta>0} \liminf_{n\rightarrow \infty} \frac{1}{n^2}\log P_n(L_n\in B(\mu,\delta)) \geq 
-H(\mu).
\end{eqnarray}
As in \cite{BenArous/Guionnet:1997} and \cite{Eichelsbacher/Stolz:2006} we can assume w.l.o.g. that

(i) $\mu$ has no atoms,\\
(ii) ${\mathcal S}:=\supp (\mu)$ is a compact subset of $\Sigma$ with ${\mathcal S}\cap 
({\mathcal N}(w)\cup \{0\})=\emptyset.$

The main idea of the proof of the lower bound, is to localise the eigenvalues in small sets, and benefit from the speed $n^2$, which provides
that the small volumes of these sets can be neglected. We will divide ${\mathcal S}$ as follows.
For $j=1,\ldots, p(n)$ let $\xi_j=\xi_j^{(n)}$ the $\frac{p(n)+1-j}{p(n)}$ quantile of $\mu$ and set  
$\xi^{(n)}=(\xi_{p(n)},\ldots,\xi_1)$, $\xi_{p(n)+1}:=\inf {\mathcal S}$ and $\xi_0=\xi_1+1$. 
Therefore, $\mu(]-\infty,\xi_{p(n)+1}])=0=\mu([\xi_{1},\infty[)$. 
Due to assumption (ii), we have 
\begin{displaymath}
-\infty<\xi_{p(n)+1}<\xi_{p(n)}<\cdots<\xi_1<\xi_0<\infty. 
\end{displaymath} 
For $\delta > 0$, $t \in \rr^{p(n)}$ write
\begin{itemize}
\item $\pi_n(t) := \{ i = 1, \ldots, p(n):\ t_i \ge 0\},\ \nu_n(t)
:= \{ 1, \ldots, p(n)\} \setminus \pi_n(t)$,
\item $I_n(\delta) :=  \{ i = 1, \ldots, p(n):\ |\xi_i^{(n)} - \xi_{i+1}^{(n)}| \le \delta \}$,
\item
$\cali_j^{(n)}(t, \delta) := [t_j - \delta, t_j + \delta] \cap \Sigma,\ j = 1,
\ldots, p(n)$, \item $$ \calj_j^{(n)}(t, \delta) := \left\{
\begin{array}{ll} ~[t_j,  t_j + \delta]& {\rm  for}\ j \in \pi_n(t),\\
~[t_j - \delta,  t_j ] & {\rm  for}\ j \in \nu_n(t),
\end{array}\right. $$
\item ${\mathbb I}_n(t, \delta) := \prod_{j=1}^{p(n)} 
\cali_j^{(n)}(t, \delta)$ 
\item
${\mathbb J}_n(t, \delta) := \prod_{j=1}^{p(n)} \calj_j^{(n)}(t,
\delta).$
\end{itemize}
We may assume that $\delta$ is fix with $0<\delta\leq \frac{1}{\theta^{\frac{1}{\theta-1}}}$.
Set 
$$
\phi_j^{(n)} := \phi_j^{(n, \delta)} := \inf\{ w_n(x):\ x \in
[\xi_{j} - \delta , \xi_{j} + \delta] \cup [\xi_{j+1}, \xi_{j-1}] 
\}, \ j = 1, \ldots, p(n),
$$ 
and analogously define $\phi_j$ when $w_n(x)$ is replaced by $w(x)$ in the above definition. 
Write $\psi_n$, resp. $\psi$ for the step function which equals $\phi_j^{(n)}$ resp. $\phi_j$ on
$]\xi_{j+1}, \xi_j]$ and is zero elsewhere
\begin{equation}\label{psi}
\psi_n=\sum_{j=1}^{p(n)} \phi_j^{(n)}1_{]\xi_{j+1}, \xi_j]}, \mbox{ resp. }
\psi=\sum_{j=1}^{p(n)} \phi_j1_{]\xi_{j+1}, \xi_j]}.
\end{equation}
Moreover,  we have: 
\begin{lemma}
\label{delta.und.n}
$\exists\ n_0 \in \N$ such that $\forall\ n>n_0: {\mathbb I}_n(\xi^{(n)}, \delta) \subset \{ x
\in \Sigma^{p(n)}:\ L_n(x) \in B(\mu, 2\delta) \},$ where
$B(\mu,2\delta)=\{\nu \in {\mathcal M}_1(\Sigma): d(\nu,\mu)\leq 2\delta\}$ with  
$$ d(\mu,\nu)=\sup_{\{f Lipschitz:\ l_f+\parallel f \parallel_{\infty}\leq1\}} \left| \int fd\mu 
-\int fd\nu \right|$$ as above. 
\end{lemma}
\vspace{-0.2cm} 
\begin{proof}{}
Let $y=(y_p(n),\ldots, y_1(n))\in {\mathbb I}_n(\xi^{(n)}, \delta)$. Thus, $y_i\in [\xi_i-\delta,\xi_i+\delta]$ and
\begin{equation}
d(L_n(y),L_n(\xi^{(n)}))
\leq \sup_{f Lip. \atop l_f+\parallel f \parallel _{\infty}\le 1} \frac{1}{p(n)}
\sum_{i=1}^{p(n)}\left|f(y_i)-f(\xi_i) \right|\leq \left| y_i-\xi_i\right|\leq \delta.\ \nonumber
\end{equation}
Applying \cite[Le. 3.3]{BenArous/Guionnet:1997} finally yields 
$d(L_n(y),\mu)\leq d(L_n(y),L_n(\xi^{(n)}))+d(L_n(\xi^{(n)}),\mu) \leq 2\delta.$
\end{proof}
Now for $n\geq n_0$, we get from Lemma \ref{delta.und.n} and ${\mathbb I}_n(\xi, \delta) \supset {\mathbb J}_n(\xi, \delta)$ 
the inequality
\begin{equation}\label{piece1} 
P_n ( L_n \in B(\mu,2\delta) ) \stackrel{Le. \ref{delta.und.n}}{\ge} Z_n Q_n ({\mathbb I}_n(\xi, \delta))
\ge \prod_{i=1}^{p(n)} \left(\phi^{(n)}_i\right)^n \int_{{\mathbb J}_n(\xi, \delta)} 
\prod_{i < j} \left|x_i^{\theta} - x_j^{\theta}\right||x_i - x_j| \ {\rm m}_{p(n)}(dx).
\end{equation}
Introducing the notation
$$ 
\R^{d,+}:=\left\{x \in \R^{d}| x_1>x_2>\ldots>x_d\right\}
$$
we focus on the last integral,
\begin{eqnarray}\label{integral}
&&\int_{{\mathbb J}_n(\xi, \delta)} \prod_{i < j} \left|x_i^{\theta} - x_j^{\theta}\right||x_i - x_j| \ 
{\rm m}_{p(n)}(dx)  \nonumber \\
&=&\int_{[- \delta, 0]^{\# \nu_n(\xi)} \times   [0, \delta]^{\# \pi_n(\xi)} }
\prod_{i < j} | (x_i+ \xi_i) - (x_j+ \xi_j)|\left|(x_i+ \xi_i)^{\theta}-(x_j+ \xi_j)^{\theta}\right| \ 
{\rm m}_{p(n)}(dx)\nonumber
\\ &\ge& \int_{ \left( [- \delta, 0]^{\# \nu_n(\xi)} \times   [0, \delta]^{\# \pi_n(\xi)}\right)
\cap \rr^{p(n), +}}
\prod_{i < j} | (x_i+ \xi_i) - (x_j+ \xi_j)| \left|(x_i+ \xi_i)^{\theta}-(x_j+ \xi_j)^{\theta}\right| \ 
{\rm m}_{p(n)}(dx) \nonumber \\
 &=& \int_{ \left( [- \delta, 0]^{\# \nu_n(\xi)} \times   [0, \delta]^{\# \pi_n(\xi)} \right)
\cap \rr^{p(n), +}} \left[
\prod_{i < j-1} |x_i- x_j  +\xi_i-\xi_j| \left|(x_i+ \xi_i)^{\theta}-(x_j+ \xi_j)^{\theta}\right| \right. \nonumber \\
&& \left.\prod_{i=1}^{p(n)-1} | x_i-x_{i+1}+ \xi_i- \xi_{i+1}| \left|(x_i+ \xi_i)^{\theta}
-(x_{i+1}+ \xi_{i+1})^{\theta}\right|\right] \ 
{\rm m}_{p(n)}(dx).
\end{eqnarray}
On $\rr^{p(n), +}$, the following inequalities hold,
\begin{eqnarray}\label{ineq1}    
\prod_{i < j-1} |x_i- x_j  +\xi_i-\xi_j|\ge \prod_{i < j-1} |\xi_i-\xi_j|
\end{eqnarray}
and as it can easily be seen (e.g. from $ab<(a+b)^2$, for $a,b>0$) 
\begin{eqnarray}\label{ineq2}
\prod_{i=1}^{p(n)-1} | x_i-x_{i+1}+ \xi_i- \xi_{i+1}|\ge \prod_{i=1}^{p(n)-1} 
|x_i-x_{i+1}|^{\frac{1}{2}}|\xi_i- \xi_{i+1}|^{\frac{1}{2}}.
\end{eqnarray}
Next, we will show that 
\begin{equation}\label{thetaineq}
\left|(x_i+ \xi_i)^{\theta}-(x_j+ \xi_j)^{\theta}\right|\ge 
\left|x_i^{\theta}+ \xi_i^{\theta}-x_j^{\theta}- \xi_j^{\theta}\right|,
\end{equation}
 where the modulus can be omitted since both terms are 
greater than 0 on $\rr^{p(n), +}$ due to our choice of the $\xi_i$.  
If $\xi_i> \xi_j>0$ and $x_i > x_j>0$, we find that      
\begin{equation} \label{bin. Satz}
(x_i+ \xi_i)^{\theta}-(x_j+ \xi_j)^{\theta}=  x_i^{\theta}+ \xi_i^{\theta}-x_j^{\theta}- \xi_j^{\theta}
+\sum_{k=1}^{\theta-1} {\theta \choose k} \overbrace{\left(x_i^k\xi_i^{\theta-k}-x_j^k\xi_j^{\theta-k}\right)}^{\ge 0}
>x_i^{\theta}+ \xi_i^{\theta}-x_j^{\theta}- \xi_j^{\theta}.
\end{equation}
In case of $0>\xi_i > \xi_j$ and $0>x_i > x_j$, we first recall that this only can happen while $\theta$ is odd.
Nevertheless, \eqref{bin. Satz} also holds: Within the sum, either $\theta-k$ is odd and $k$ is even or vice versa. 
Hence, if $k$ is even $0>\xi_i^{\theta-k}x_i^k > \xi_j^{\theta-k}x_i^k
> \xi_j^{\theta-k}x_j^k$, while the other case follows in an analogous way.
The last case to consider is $\xi_i>0>\xi_j$ and $x_i>0>x_j$. This case also only comes up for $\theta$ being odd, and
we find $x_i^k\xi_i^{\theta-k}>0$ and $-x_j^k\xi_j^{\theta-k}>0$ since either $k$ is even or $\theta-k$.
Thus, we have proven \eqref{thetaineq} and  this can on one hand further be used to obtain either
\begin{equation}\label{ineq3}
x_i^{\theta}+ \xi_i^{\theta}-x_j^{\theta}- \xi_j^{\theta} > \left(x_i^{\theta}-x_j^{\theta}\right)^{\frac{1}{2}}
\left(\xi_i^{\theta}-\xi_j^{\theta}\right)^{\frac{1}{2}}
\end{equation}
or on the other hand yield
\begin{equation}\label{ineq4}
x_i^{\theta}+ \xi_i^{\theta}-x_j^{\theta}- \xi_j^{\theta} > \xi_i^{\theta}-\xi_j^{\theta},
\end{equation}
since $x_i^{\theta}>x_j^{\theta}$ on $\R^{p(n),+}$ for $\theta \in \N$.\\
Putting together \eqref{ineq1}, \eqref{ineq2}, \eqref{ineq3} and \eqref{ineq4}, we can bound 
\eqref{integral} from below by
\begin{eqnarray}\label{letztint}
&&\prod_{i<j-1}\left[(\xi_i-\xi_j)\left(\xi_i^{\theta}-\xi_j^{\theta}\right)\right]\prod_{i=1}^{p(n)-1}
\left[(\xi_i-\xi_{i+1})^{\frac{1}{2}}\left(\xi_i^{\theta}-\xi_{i+1}^{\theta}\right)^{\frac{1}{2}}\right]\times\nonumber \\
&&\int_{\left( [- \delta, 0]^{\# \nu_n(\xi)} \times   [0, \delta]^{\# \pi_n(\xi)} \right)
\cap \rr^{p(n), +}}\prod_{i=1}^{p(n)-1}
\left[(x_i-x_{i+1})^{\frac{1}{2}}\left(x_i^{\theta}-x_{i+1}^{\theta}\right)^{\frac{1}{2}}\right]{\rm m}_{p(n)}(dx).
\end{eqnarray}
Applying the mean value theorem for $\theta\ge 2$ (while for $\theta$=1 the following inequality will be an
equality) now yields $x_i^{\theta}-x_{i+1}^{\theta}=(x_i-x_{i+1})\theta \zeta^{\theta-1},$
with $\zeta \in [x_{i+1},x_i].$ Since $|x_i|\le \delta\ \forall\ i$ and  $\delta<\frac{1}{\theta^{\frac{1}{\theta-1}}}$,
we easily see that $0<\theta \zeta^{\theta-1}<1$. 
Therefore, $0<x_i^{\theta}-x_{i+1}^{\theta}\le x_i-x_{i+1}$ and we obtain 
\begin{eqnarray}
&&\int_{\left( [- \delta, 0]^{\# \nu_n(\xi)} \times   [0, \delta]^{\# \pi_n(\xi)} \right)
\cap \rr^{p(n), +}}\prod_{i=1}^{p(n)-1}
\left[(x_i-x_{i+1})^{\frac{1}{2}}\left(x_i^{\theta}-x_{i+1}^{\theta}\right)^{\frac{1}{2}}\right]{\rm m}_{p(n)}(dx)
\nonumber \\
&\ge& \int_{\left( [- \delta, 0]^{\# \nu_n(\xi)} \times   [0, \delta]^{\# \pi_n(\xi)} \right)
\cap \rr^{p(n), +}}\prod_{i=1}^{p(n)-1}\left(x_i^{\theta}-x_{i+1}^{\theta}\right){\rm m}_{p(n)}(dx)\nonumber \\
&=&\frac{1}{\theta^{p(n)}}\int_{\left( [- \delta^{\theta}, 0]^{\# \nu_n(\xi)} \times    
[0, \delta^{\theta}]^{\# \pi_n(\xi)} \right)\cap \rr^{p(n), +}}\prod_{i=1}^{p(n)-1}\left(x_i-x_{i+1}\right)
\prod_{i=1}^{p(n)}|x_i|^{\frac{1-\theta}{\theta}}{\rm m}_{p(n)}(dx) \label{end1}\\
&\ge&\frac{1}{\theta^{p(n)}}\int_{\left( [- \delta^{\theta}, 0]^{\# \nu_n(\xi)} 
\times   [0, \delta^{\theta}]^{\# \pi_n(\xi)} \right)
\cap \rr^{p(n), +}}\prod_{i=1}^{p(n)-1}\left(x_i-x_{i+1}\right){\rm m}_{p(n)}(dx). \label{end2}
\end{eqnarray}
The equality \eqref{end1} is a simple application of the transformation formula, while the inequality \eqref{end2} holds, 
because $\theta$ is an integer and $\delta<1$, so that $|x_i|^{\frac{1-\theta}{\theta}}\ge 1$.
For the last integral, we substitute $u_{p(n)}=x_{p(n)}, u_{i-1}=x_{i-1}-x_{i}, i=p(n),\ldots,2$ and obtain 
\begin{eqnarray}\label{piece2}
\int_{\left( [- \delta^{\theta}, 0]^{\# \nu_n(\xi)} \times   [0, \delta^{\theta}]^{\# \pi_n(\xi)} \right)
\cap \rr^{p(n), +}}\prod_{i=1}^{p(n)-1}\left(x_i-x_{i+1}\right){\rm m}_{p(n)}(dx)
=\int_{U_n}\prod_{i=2}^{p(n)} u_i\, {\rm m}_{p(n)}(dx),\nonumber\\
\ge\int_{\left[0,\frac{\delta^{\theta}}{p(n)}\right]^{p(n)}}\prod_{i=2}^{p(n)} u_i\ {\rm m}_{p(n)}(dx)= 
\left(\frac{1}{2}\right)^{p(n)-1}\left(\frac{\delta^{\theta}}{p(n)}\right)^{2(p(n)-1)+1},
\end{eqnarray}
where 
$
U_n=\left\{u=(u_{p(n)},\ldots,u_1)|u_{p(n)} 
\in[-\delta^{\theta},\delta^{\theta}],u_i\in [0,\delta^{\theta}], i=1,\ldots,p(n)-1:
|\sum_{i=1}^{p(n)} u_i|\le \delta^{\theta} \right\}.
$ 
From \eqref{piece1}, \eqref{letztint} and \eqref{piece2} we get,
\begin{eqnarray}
\hspace{-0.15in}\frac{1}{n^2}\log P_n ( L_n \in B(\mu,2\delta) )
 &\ge& \label{term1.1}
\left(\frac{p(n)}{n^2}\right)^2 \frac{1}{p(n)^2} \sum_{i<j-1}\log\left(\left[(\xi_i-\xi_j)
\left(\xi_i^{\theta}-\xi_j^{\theta}\right)\right]\right) \\
&+&\left(\frac{p(n)}{n^2}\right)^2 \frac{1}{2p(n)^2}\sum_{i=1}^{p(n)-1} \log\left(
\left[(\xi_i-\xi_{i+1})\left(\xi_i^{\theta}-\xi_{i+1}^{\theta}\right)\right]\right) \label{term1.2}\\
&+&\frac{1}{n} \sum_{i=1}^{p(n)} \log \phi^{(n)}_i \label{term2}\\
&+&\frac{1}{n^2} \log \left(2\left(\frac{1}{2\theta}\right)^{p(n)} \right)+\frac{2p(n)-1}{n^2} 
( \log \delta^{\theta}- \log p(n) ). \label{Nullterm}
\end{eqnarray}
It is easily seen that \eqref{Nullterm} converges to 0 as $n$ tends to infinity.\\
Now, for (\ref{term1.1}) and (\ref{term1.2})   observe that for $\theta=1, \ \mbox{ resp. } \theta \in \N$,
\begin{eqnarray*}
&& \int_{x < y} \log(y^{\theta} - x^{\theta}) \mu(dx) \mu(dy) -\int_{(x, y) \in [\xi_{p(n)+1}, \xi_{p(n)}]^{\times 2}}\log|y^{\theta} - x^{\theta}| \mu(dx) \mu(dy)\\
&=& \sum_{i < j} \int_{(x, y) \in [\xi_{j+1}, \xi_j] \times [\xi_{i+1}, \xi_i]} \log(y^{\theta} - x^{\theta}) 
\mu(dx) \mu(dy) \\
&&
+ \frac{1}{2}\sum_{i=1}^{p(n)-1} \int_{(x, y) \in [\xi_{i+1}, \xi_i]^{\times 2}}\log|y^{\theta} - x^{\theta}| 
\mu(dx) \mu(dy)\\  
&\le& \frac{1}{p(n)^2} \sum_{i<j} \log(\xi_i^{\theta} -
\xi_{j+1}^{\theta}) + \frac{1}{2 p(n)^2} \sum_{i=1}^{p(n)-1} \log 
(\xi_i^{\theta} - \xi_{i+1}^{\theta}).
\end{eqnarray*}
Because $\xi_{p(n)+1}$ was defined to be the infimum of the support $\mathcal S$ and we look at measures without atoms, we
observe that $\xi_{p(n)}\rightarrow \xi_{p(n)+1}= \inf \mathcal S$ for $n\rightarrow \infty$. Thus, we find
\begin{eqnarray*}\label{fehl}
\lim_{n\rightarrow \infty} \int_{(x, y) \in [\xi_{p(n)+1}, \xi_{p(n)}]^{\times 2}}\log|y^{\theta} - x^{\theta}| \mu(dx) \mu(dy)
= \int_{\emptyset}\log|y^{\theta} - x^{\theta}| \mu(dx) \mu(dy)
=0.
\end{eqnarray*}
We have shown that the limit of \eqref{term1.1} and \eqref{term1.2} can be bounded from below in the following way,
\begin{eqnarray}\label{338}
&& \int \int \log(y^{\theta} - x^{\theta}) \mu(dx) \mu(dy)
=\frac{1}{2} \int_{x < y} \log(y^{\theta} - x^{\theta}) \mu(dx) \mu(dy) \nonumber \\
&\le&\frac{1}{2} \lim_{n \rightarrow \infty}\left(\frac{1}{p(n)^2} \sum_{i<j} \log(\xi_i^{\theta} -
\xi_{j+1}^{\theta}) + \frac{1}{2 p(n)^2} \sum_{i=1}^{p(n)-1} \log 
(\xi_i^{\theta} - \xi_{i+1}^{\theta})\right).
\end{eqnarray}
As to \eqref{term2}, observe that 
$$
 \frac{1}{n} \sum_{j=1}^{p(n)} \log \phi_j^{(n)} = \frac{p(n)}{n} \int \log \psi_n d\mu,
$$
where $\psi_n$ was defined in \eqref{psi}. 
Denote  by $l$ a Lipschitz constant of $\log w$ on $\cals$ according to (a1.3). For $\eta > 0$   
write 
$$
l^{\eta} := \max\{|\log w(x) - \log w(y)|:\ |x-y| \le
\eta\}.
$$
Note that $l^{\eta} \le l \eta$, again because of the Lipschitz continuity (a1.3) and since 
$\mathcal{N}(w) \cap \mathcal{S}=\emptyset$. Next, define
$
M(n, \delta) := \bigcup_{j \in I_n(\delta)} [\xi_{j+1}^{(n)},
\xi_j^{(n)}]
$ 
and
$
C := \max\{ \log w(x):\  x \in \cals \} -
\min\{ \log w(x):\ x \in \cals \},
$
where assumption (ii) provides that $C$ is finite. Let $\epsilon > 0$. Since for all $n$ and $j \ge 1$ one has
$\mu\left([\xi_{j+1}^{(n)}, \xi_j^{(n)}]\right) = 1 / p(n)$, and
since $\delta $ is fixed, one has $\mu( M(n, \delta)^c) \le
\epsilon$ for large $n$. 
Now let $n$ be large enough such that one also has 
$\| \log w_n - \log w\|_{\infty} \le \epsilon$, which is possible because of (a1.2). Then
\begin{eqnarray*}
 \int |\log \psi_n - \log w| d\mu 
&\le& \sum_{j \in I_n(\delta)} \int_{[\xi_{j+1}, \xi_j]} \left\{|\log \psi -\log w| + \epsilon \right\} d\mu  \\
&&+ \sum_{j \in I^c_n(\delta)} \int_{[\xi_{j+1}, \xi_j]} \left\{|\log \psi -\log w| + \epsilon \right\} d\mu + \epsilon\\
&\le& p(n) \frac{1}{p(n)}\ (l^{2\delta}+\epsilon)  + (C+ \epsilon)\epsilon 
\le 2l \delta+ (C+\epsilon+1)\epsilon
\end{eqnarray*}
This yields 
\begin{equation}\label{339}
 \limsup_{n \to \infty} \left| \frac{1}{n}
\sum_{j=1}^{p(n)} \log \phi_j^{(n)} - \kappa \int \log w\
d\mu\right| = O(\delta).
\end{equation}
Putting together \eqref{term1.1}-\eqref{339}  
yields the lower bound \eqref{lowbound}, whenever the empirical measure lies in a ball around some measure $\mu$ in $\mathcal{M}_1(\Sigma)$.
\subsection{Extending the weak LDP into a full one}
As $\Sigma$ is a closed subset of the Polish space $\R$, $\calm_1(\Sigma)$ is a Polish space, cf. \cite[Thm. 6.2]{Parthasarathy:ProbMeasure} and since the $B(\mu,\delta)$ are a base for the weak topology on $\calm_1(\Sigma)$, we can
apply \cite[Thm. 4.1.11]{Dembo/Zeitouni:LargeDeviations}, which provides
a weak LDP for $(P_n\circ L_n^{-1})_n$ on $\calm_1(\Sigma)$ with rate function $H$ and speed $n^2$.
Because $(P_n \circ L_n^{-1})_n$ is also exponentially tight, Lemma 1.2.18 in \cite{Dembo/Zeitouni:LargeDeviations} extends the weak LDP into a full one. 
But our original aim was to derive a LDP for $(Q_n \circ L_n^{-1})_n$.
Setting $A = G = \calm_1(\Sigma)$ in the lower and upper bound at the LDP for $(P_n\circ L_n^{-1})_n$, 
we obtain  
$$ \lim_{n \to \infty} \frac{1}{n^2} \log Z_n = - \inf_{\mu \in
\calm_1(\Sigma)} \int F d\mu^{\otimes 2}.$$ 
The right hand side has been shown to be $<+ \infty$ by Lemma \ref{bounded} (ii) and establishes \eqref{qv12}. 
Now,
$$
 \frac{1}{n^2} \log Q_n(L_n \in A) = \frac{1}{n^2} \left( \log P_n(L_n \in A) - \log Z_n\right)
$$
for any Borel set $A$ in $\calm_1(\Sigma)$. Hence Theorem \ref{ldp} is proven.

\section{Sketch of the proof of Theorem \ref{ldpMOP}}
\label{mopproof}
The Vandermonde-like products $\Delta(X)$ and $\Delta(X,Y)$ in \eqref{angelesco} lead to the observation,
that to overcome the singularity of the logarithm at several places is the most delicate part. Therefore, we
will skip most of the technical details of the proof worrying about too many replications of our arguments. 
Instead we give a sketch to be able to discover the rate function. Consider a neighborhood $O$ of $\mu=(\mu_1, \ldots, \mu_p)
\in \cM_1(\Gamma_{p, \vec{n}})$. With $\tilde{O} := \{ x \in {\Bbb R}^n: L_n \in O \}$
we obtain
$$
Q_n(L_n \in O) = \frac{1}{Z_n} \int \cdots \int_{\tilde{O}} \exp( F(x)) d x_1 \ldots d x_n
$$
with
\begin{equation} \label{function1}
F(x) = \sum_{i=1}^p \sum_{1 \leq j<k\leq n_i} \log (x_j^{(i)} - x_k^{(i)})^2 + \sum_{1 \leq i < j \leq p} \sum_{k=1}^{n_i} \sum_{l=1}^{n_j} \log(x_k^{(i)} - x_l^{(j)})
- n \sum_{i=1}^p \sum_{k=1}^{n_i} V_i(x_k^{(i)}).
\end{equation}
Very roughly speaking, we consider the case when $O$ goes to a point $\mu$. Then for  $x \in \tilde{O}$ the first summand in \eqref{function1}
can be approximated by
$$
 \sum_{i=1}^p \frac{n_i^2}{2} \int \int \log |x-y|^2 d\mu_i(x) d\mu_i(y),
$$
the second summand by
$$
\sum_{1 \leq i < j \leq p} n_i n_j \int \int \log|x-y| d\mu_i(x) d\mu_j(y)
$$
and the third summand by
$n\,  \sum_{i=1}^p n_i \int V_i(x) d\mu_i(x)$. 
With \eqref{c1} one gets
\begin{eqnarray*}
& & \frac{1}{n^2} \log Q_n (L_n \in O) \\
& \approx & \sum_{i=1}^p \frac{r_i^2}{2} \int \int \log |x-y|^2 d\mu_i(x) d\mu_i(y) +  \sum_{1 \leq i < j \leq p} r_i r_j \int \int \log|x-y| d\mu_i(x) d\mu_j(y) \\ &+&
\sum_{i=1}^p r_i \int V_i(x) d\mu_i(x) + \frac{1}{n^2} \log Z_n.
\end{eqnarray*}



\newcommand{\SortNoop}[1]{}\def\cprime{$'$}
\providecommand{\bysame}{\leavevmode\hbox to3em{\hrulefill}\thinspace}
\providecommand{\MR}{\relax\ifhmode\unskip\space\fi MR }
\providecommand{\MRhref}[2]{%
  \href{http://www.ams.org/mathscinet-getitem?mr=#1}{#2}
}
\providecommand{\href}[2]{#2}


\end{document}